\documentclass[10 pt, conference]{IEEEtran}
\IEEEoverridecommandlockouts                       

\usepackage{url}
\usepackage[utf8]{inputenc}
\usepackage[T1]{fontenc}
\usepackage{amsmath,amssymb,amsfonts}
\usepackage{graphicx}
\usepackage{epstopdf}
\usepackage{subfigure}
\graphicspath{{\figures}}
\usepackage{textcomp}
\usepackage{xcolor}
\usepackage{verbatim}
\usepackage{makecell}
\usepackage{booktabs}
\usepackage{cite}
\usepackage{caption}
\usepackage{extarrows}
\usepackage{bm}
\usepackage{lettrine}
\usepackage[implicit=false]{hyperref}
\usepackage{amsthm}
\usepackage{algorithm,algorithmic}
\pdfoutput=1

\begin{document}
\title{\LARGE \bf
A Novel Two-Layer Codebook Based Near-Field Beam Training for Intelligent Reflecting Surface
\vspace{-0.5em}} 

\author{
\IEEEauthorblockN{Tao Wang\IEEEauthorrefmark{1},
Jie Lv\IEEEauthorrefmark{1},
Haonan Tong\IEEEauthorrefmark{1}, 
Changsheng You\IEEEauthorrefmark{2}, and
Changchuan Yin\IEEEauthorrefmark{1},~\IEEEmembership{Senior Member IEEE}}
\IEEEauthorblockA{\IEEEauthorrefmark{1}School of Information and Communication Engineering, Beijing University of Posts and Telecommunications, Beijing, China}
\IEEEauthorblockA{\IEEEauthorrefmark{2}Department of Electrical and Electronic Engineering, Southern University of Science and Technology, Shenzhen, China}
\IEEEauthorblockA{Emails:$^*$\{taowang, lvj, hntong, ccyin\}@bupt.edu.cn, $\dag$youcs@sustech.edu.cn}
\vspace{-1.5em}
}
\maketitle
\thispagestyle{empty}
\pagestyle{empty}

\begin{abstract}
In this paper, we study the codebook-based near-field beam training for intelligent reflecting surfaces (IRSs) aided wireless system. In the considered model, the near-field beam training is critical to focus signals at the location of user equipment (UE) to obtain prominent IRS array gain. However, existing codebook schemes cannot achieve low training overhead and high receiving power simultaneously. 
To tackle this issue, a novel two-layer codebook based beam training scheme is proposed. The layer-1 codebook is designed based on the omnidirectionality of a random-phase beam pattern, which estimates the UE distance with training overhead equivalent to that of one DFT codeword. Then, based on the estimated UE distance, the layer-2 codebook is generated to scan candidate UE locations and obtain the optimal codeword for IRS beamforming. Numerical results show that compared with benchmarks, the proposed two-layer beam training scheme achieves more accurate UE distance and angle estimation, higher data rate, and smaller training overhead.

\end{abstract}


\section{Introduction}
Intelligent reflecting surfaces (IRSs) are expected to be a key technology in next-generation wireless networks due to their potential to reconfigure the wireless environment in a cost-effective manner\cite{c8, c4, csyou_channel-estimation-passive-beamforming}. 
However, due to the nature in terms of signal reflection, large-scale IRS is required to capture enough energy and achieve significant beamforming gain \cite{c3}. In practice, to enhance the local coverage, the large-scale IRS can be placed near the user equipment (UE), while establishing a reliable line-of-sight (LoS) dominant link with the base station (BS)\cite{intelligent-communication-environments-enabled-by-Ultra-Massive-MIMO}. This thus leads to a fundamental paradigm shift from the far-field IRS-aided wireless communications to the near-field counterpart.

Beam training in near-field scenarios poses several new challenges, such as the degraded training accuracy and the drastic increase of training overhead, which have been addressed in the literature \cite{near-field-V1,  spherical-Wave-Channel, Channel-Estimation-for-XL-MIMO, Fast-Near-Field-Beam-Training-IRS}. Specifically, the near-field channel is studied in \cite{spherical-Wave-Channel}, where a polar-domain (angle plus distance) channel model was proposed, showing that near-field beamforming should focus signal energy at the UE location. Traditional far-field beam training, such as the discrete Fourier transform~(DFT) codebook based beam training, only tries to match the angle of the UE, thus resulting in a significant decrease of beamforming gain in the near-field scenarios. To tackle this issue, our previous work in \cite{near-field-V1} intuitively added a distance-domain beam training after the DFT beam training, which outperforms the DFT codebook scheme. Furthermore, a polar-domain codebook dedicated to the near-field channel was proposed in \cite{Channel-Estimation-for-XL-MIMO}, where the polar-domain channel space is sampled based on compressed sensing theory, and an exhaustive beam search strategy is adopted. To reduce the prohibitively high beam training overhead as that in \cite{Channel-Estimation-for-XL-MIMO}, a two-phase beam training method was proposed in \cite{Fast-Near-Field-Beam-Training-IRS}, where candidate UE angles are obtained from the DFT codebook based beam training in the first phase, and a shortlisted polar-domain codeword from \cite{Channel-Estimation-for-XL-MIMO} is adopted to estimate the UE distance in the second phase.

However, we observe that near-field beam patterns under far-field beamforming are broadened in the near-field region, resulting in inaccurate DFT codebook based angle estimation to be insufficient, and hence performance degradation of the current near-field IRS codebook based beam training schemes \cite{near-field-V1, Fast-Near-Field-Beam-Training-IRS}. 
Motivated by this observation, we propose a novel two-layer codebook to achieve both low training overhead and high IRS array gain. The key contributions of this paper are summarized as follows. 
\begin{itemize}
\item We elaborate on the problem caused by using far-field beamforming in the near-field region, based on the considered model of IRS-assisted wireless systems.
\item To achieve low training overhead and high IRS array gain, we propose a novel two-layer codebook based beam training scheme. The layer-1 codebook is composed of independent random-phase vectors, which generates an omnidirectional IRS beam for the estimation of UE distance. The layer-2 codebook is generated based on the estimated UE distance from the first layer, which searches the candidate UE locations and finally obtains the optimal codeword for IRS beamforming.
\item We provide theoretical analysis and numerical simulations for the proposed beam training scheme. The numerical results show that, compared to benchmarks, the proposed scheme provides more accurate estimations of the distances and angles of the UE, achieving higher data rates with a smaller training overhead.
\end{itemize}

The rest of the paper is organized as follows. In Section II, we illustrate the system model. In Section III, through the elaboration of near-field beam patterns, we provide insights of near-field beam training. In Section IV, we propose a novel two-layer codebook based near-field beam training scheme, which is theoretically analysed. In Section V, extensive numerical results are provided. Finally, we draw conclusions in Section VI.

	\begin{figure}[ht]
	\centering
	\includegraphics[scale = 0.7]{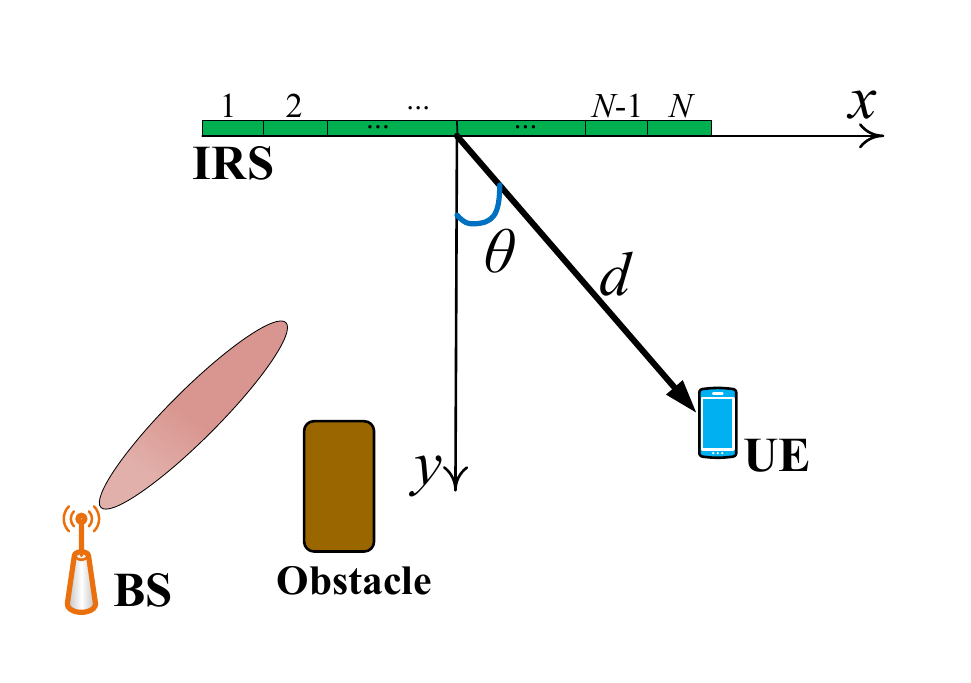}
 \vspace{-1em}
	\caption{An IRS-assisted wireless system.}
	\label{fig:system}
	\vspace{-0.5em}
	\end{figure}
\vspace{1em}

\section{System Model}
We consider the downlink beam training of an IRS-assisted narrow-band wireless system, as shown in Fig. \ref{fig:system}. The system consists of an $M$-antenna BS, a single-antenna UE, and an $N$-element uniform linear array (ULA) IRS with $\lambda/2$ spacing, where $\lambda$ is the carrier wavelength\footnote{The ULA IRS is used in this paper for the simplicity of analysis. However, the proposed beam training scheme can be readily extended to Uniform Planar Array~(UPA) IRS by also using independent random-phase beamforming vectors for the proposed layer-1 codebook, and using the spatial angles of the two dimensional DFT codebook as the candidate UE angles to generate the proposed layer-2 codebook. The details are illustrated in Section IV.}. Let $D$ denotes the length of the IRS. We assume that the UE is located in the near-field region of the IRS, and there only exists a virtual LoS BS-IRS-UE link due to the occlusion of the direct BS-UE link. The IRS is placed on the $x$-axis with its center at the origin. Specifically, the location of the $n$-th IRS element is $\big((n-\frac{N+1}{2})\frac{\lambda}{2}, 0\big)$, where $N$ is assumed to be an even number. The UE's location is $(d\sin{\theta}, d\cos{\theta})$, where $\theta$ and $d$ represent the angle and distance of the UE from the origin, respectively. Let $\mathbf{Euc}(\cdot,\cdot)$ denotes the Euclidean distance between two coordinates. The distance between the $n$-th IRS element and the UE, denoted by $d_n$, is given by:
\begin{equation}
 \label{euc_distance}
 \begin{split}
 d_n(d,\theta) &= \mathbf{Euc}\left(\big((n-\frac{N+1}{2})\frac{\lambda}{2}, 0\big),(d\sin{\theta}, d\cos{\theta})\right).
 \end{split}
\end{equation}

\underline{Channel model:} We consider the propagation environment with limited scattering, which is typical for mmWave channels. Since the BS and IRS are at fixed locations once deployed, the BS-IRS link can be considered
as quasi-static. As such, we assume for simplicity that the BS has aligned its transmit beamforming with the BS-IRS LoS channel and thus can be treated as having an equivalent single antenna\cite{multi-beam-training-IRS}\footnote{For other BS-IRS channel settings, the design of the BS precoding during IRS beam training is an interesting topic but is left for our future work.}. Then, the effective BS-IRS channel, denoted by $\boldsymbol{f} \in \mathbb{C}^{N\times1}$, can be modeled as $\boldsymbol{f} = f\left[ e^{j\omega_1}, e^{j\omega_2}, \cdots, e^{j \omega_N} \right]^{T}$, where $f$ denotes the complex-valued path gain of the BS-IRS
link; $\omega_n$ denotes the phase of channel between the BS and the $n$-IRS element; $j$ denotes the imaginary unit. Let
$\boldsymbol{h^H}\in \mathbb{C}^{N\times 1}$ denote the near-field IRS-UE channel, which can be modeled as \cite{spherical-Wave-Channel}
  \begin{equation}
 \label{h-model}
 \begin{split}
 [\boldsymbol{h^H}]_{n}  = \sqrt{\frac{G^{\rm U} A^{\rm U} G^{\rm I}}{4\pi d_n^2}}e^{-j\frac{2\pi d_n}{\lambda}},
 \end{split}
 \end{equation}
where $G^{\rm U}$, $A^{\rm U}$, and $G^{\rm I}$ denote the receiving gain of the UE antenna, the effective aperture of the UE antenna, and the gain of each IRS element, respectively.

Let $\boldsymbol{\Phi}\triangleq diag( e^{j\varphi_1}, \cdots, e^{j\varphi_N})$ denotes the adjustable phase shifts introduced by the IRS. The received signal at the UE can be represented as
\begin{equation}
 \label{received_signal}
 \begin{split}
 y &= \boldsymbol{h}^{H}\boldsymbol{\Phi}\boldsymbol{f}s+z,
 \end{split}
 \end{equation}
where $s$ denotes the symbol transmitted by the BS with power $P_A$, $z \sim{\mathcal{C}\mathcal{N}(0, \sigma^2)}$ denotes the received Additive Gaussian white noise~(AWGN) with zero mean and variance $\sigma^2$.
Substituting (\ref{h-model}) into (\ref{received_signal}) yields the received signal power represented as
 \begin{equation}
 \label{received-power}
 \begin{aligned}
  P_{\rm r} =& ~\frac{P_{\rm A}G^{\rm U} A^{\rm U} G^{\rm I}f^2}{4\pi}\Big|\sum_{n=1}^{N}\frac{1}{ d_n}e^{j(\varphi_n+\omega_n+\frac{2\pi d_n}{\lambda})}\Big|^2.
  \end{aligned}
 \end{equation}
The achievable rate in bits/second/Hz (bps/Hz) can be obtained as
 \begin{equation}
 \label{achievable-rate}
 \begin{aligned}
  R =& \log_2(1+\frac{P_{\rm r}}{\sigma^2}).
  \end{aligned}
 \end{equation}

\section{Near-Field Beam Patterns}

It can be easily obtain from (\ref{received-power}) that the optimal near-field IRS beamforming is \cite{c3}:
\begin{equation}
 \label{IRS_phase_config}
 \begin{split}
 \varphi^{\rm {near}}_n(\theta,d)= -\omega_n - \frac{2\pi d_n}{\lambda},~~n = 1, \cdots, N,
 \end{split}
\end{equation}
where the in-phase superposition of signals from all IRS elements is realized at the location of UE and the maximal received power can be obtained as $P_{\rm max} = P_{\rm A}f^2\left|\sum_{n=1}^{N}\sqrt{\frac{G^{\rm U} A^{\rm U} G^{\rm I}}{4\pi d_n^2}}\right|^2$. Since $\omega_n$ is known, the optimal $\boldsymbol{\Phi}$ is only determined by $d_n$, and thus determined by $\theta$ and $d$ according to (\ref{euc_distance}). Therefore, the location of the UE should be obtained to perform the optimal near-field IRS beamforming.
By contrast, the conventional far-field IRS beamforming only based on UE angle is:
\begin{equation}
 \label{IRS_far_field_phase_config}
 \begin{split}
 \varphi^{\rm {far}}_n(\theta)= -\omega_n-\pi(n-1)\sin{\theta},~~n = 1, \cdots, N.
 \end{split}
\end{equation}

Comparing (\ref{IRS_phase_config}) and (\ref{IRS_far_field_phase_config}), it is observed that the far-field beamforming induces phase errors in the near-field channel condition since it does not take into account the UE distance, which inevitably leads to the decrease of the IRS array gain in the near-field region. The closed-form expression of the array gain loss caused by the far-field beamforming in the near-field region has been derived in our previous work\cite{near-field-V1}. To clearly show the difference between the near-field beamforming and far-field beamforming in the near-field region, we make the following definition of the near-field beam pattern.
\newtheorem{definition}{Definition}

\begin{definition}
The near-field beam pattern is defined as the signal power measured on the circle around the IRS center and normalized by the achievable maximal signal power, where the radius of the circle, $d_{\rm p}$, is less than the IRS's Rayleigh distance, $2D^2/\lambda$. 
\end{definition}

Similar to the derivation of (\ref{received-power}), the normalized near-field beam pattern is represented as:
\begin{equation}
\label{measured_signal_power}
\begin{split}
&P_{\rm p}(\varphi_n) = \frac{\Big|\sum_{n=1}^{N}\frac{1}{d^{'}_n}e^{j(\varphi_n+\omega_n+\frac{2\pi d^{'}_n}{\lambda})}\Big|^2}{\Big|\sum_{n=1}^{N}\frac{1}{ d^{'}_n}\Big|^2}, \theta_{\rm p}\in [-\frac{\pi}{2}, \frac{\pi}{2}),
\end{split}
\end{equation}
where $d^{'}_n \triangleq \mathbf{Euc}\left(\big((n-\frac{N+1}{2})\frac{\lambda}{2}, 0\big),(d_{\rm p}\sin{\theta_{\rm p}}, d_{\rm p}\cos{\theta_{\rm p}})\right)$ denotes the distance between the $n$-th IRS element and the measuring point with coordinate $(d_{\rm p}\sin{\theta_{\rm p}}, d_{\rm p}\cos{\theta_{\rm p}})$. By substituting $\varphi_n$ in (\ref{measured_signal_power}) with $\varphi^{\rm {near}}_n(\theta,d)$ and $\varphi^{\rm {far}}_n(\theta)$, respectively, the near-field beam patterns with $d_{\rm p} = d$ under near-field and far-field beamforming can be obtained, which are shown in Fig. \ref{fig:beam_patterns} with $\theta = 0^\circ,~ 20^\circ,~ 40^\circ,~ 60^\circ$ and $d =$ 6 m, under the system setup with $D = 1.28$m, $N = 256$, and $\lambda = 0.01$m.
\begin{figure}[ht]
    \vspace{-1em}
    \centering
    \includegraphics[scale = 0.6]{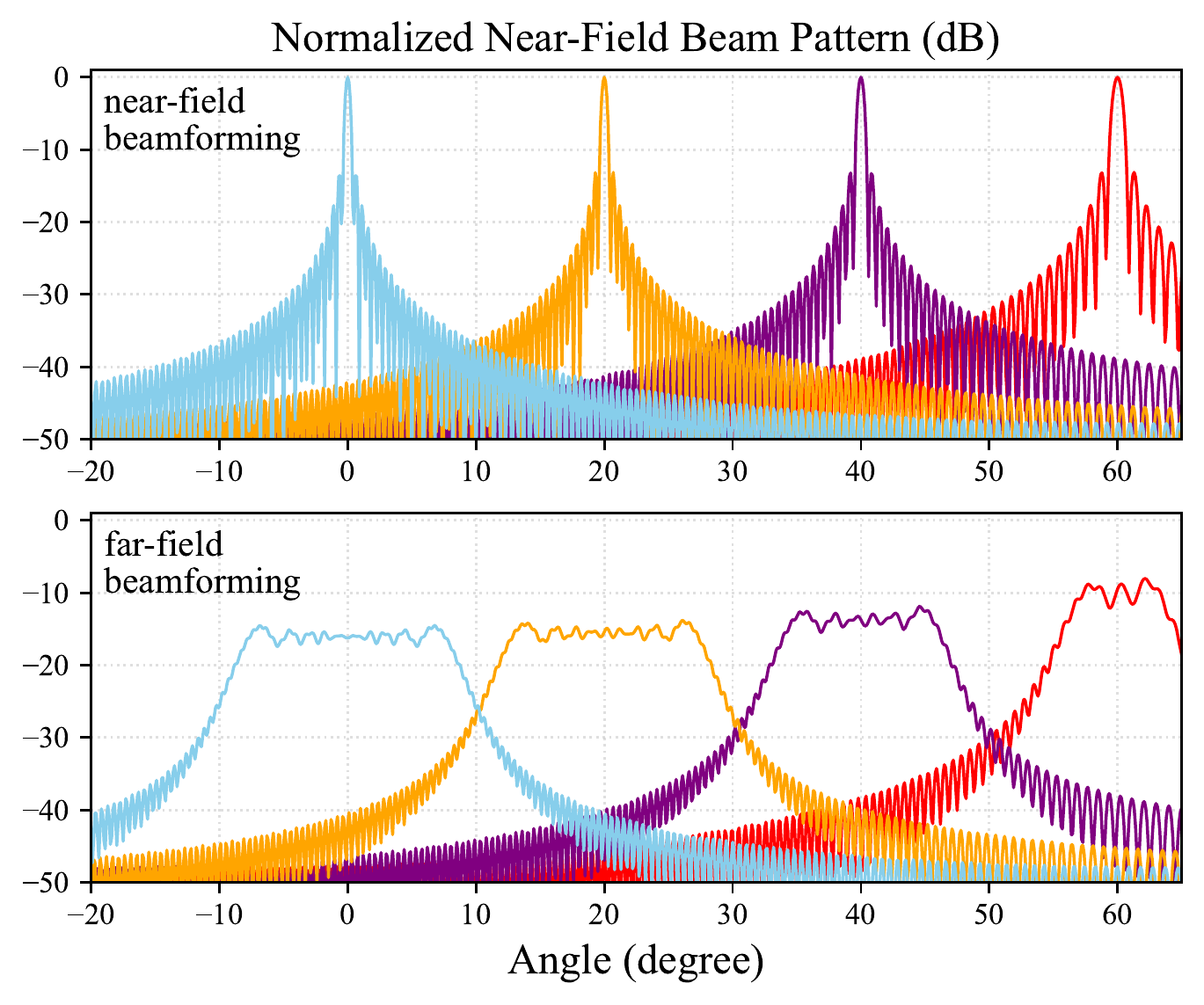}
    \vspace{-1em}
    \caption{Normalized near-field beam patterns under far-field and near-field beamforming.}
    \label{fig:beam_patterns}
\end{figure}

\noindent\textbf{Observation 1}~(beam broadening effect)\textbf{.} In Fig. \ref{fig:beam_patterns}, the near-field beam patterns under far-field beamforming get broadened and lose the single peak characteristic, which is called the beam broadening effect in this paper. In this case, the accuracy of the angle estimation under far-field beam sweeping \cite{near-field-V1, Fast-Near-Field-Beam-Training-IRS} cannot be guaranteed, resulting in performance degradation, as shown in the numerical results of Section IV. By contrast, the near-field beam patterns under near-field beamforming maintain the single peak characteristic, which reveals that if the UE distance $d$ is given, the UE angle $\theta$ can be accurately estimated through the near-field beam sweeping with focuses on the circle with radius $d$.
Inspired by this observation, a novel two-layer codebook based beam training scheme is proposed to estimate $d$ and $\theta$ successively. This scheme can accurately estimate the angle of the UE through near-field beam sweeping, which is not possible under far-field beam sweeping due to the beam broadening effect.

\section{Proposed Two-Layer Near-Field Beam Training}
In this section, we propose a novel two-layer codebook based IRS beam training scheme, which estimates $d$ and $\theta$ of UE successively.
To estimate $d$, we make an approximation of equation (\ref{h-model}) as follows:
\begin{equation}   
 \label{h-approximation}
 \begin{split}
 [\boldsymbol{h}]_{n}  \approx 
 \underbrace{\sqrt{\frac{G^{\rm U} A^{\rm U} G^{\rm I}}{4\pi d^2}}}_{b}e^{j\frac{2\pi d_n}{\lambda}},
 \end{split}
 \end{equation}
 where $d_n, n = 1, \cdots, N$ are approximated as $d$ to calculate the amplitude $b$.
 To determine when this approximation holds, we compute the distance differential from the IRS elements to the UE between the center and the edge of the IRS. For example, suppose $\theta = 0^{\circ}$ (the analysis can be generalized to arbitrary UE angles). With the equal-distance approximation in (\ref{h-approximation}), the signal from each IRS element travels the same distance, $d$, to the UE. But at the edge of the IRS, it has actually traveled the distance $\sqrt{d^2 + D^2/4}$. The error ratio of the approximation is $\frac{d-\sqrt{d^2 + D^2/4}}{\sqrt{d^2 + D^2/4}}$. With $D$ = 1.5 m, $d$ = 5 m, the ratio is about 1.14\%, which has a minor impact. With (\ref{h-approximation}), equation (\ref{received-power}) can be approximated as
 \begin{equation}   
 \label{received-power-approximation}
 \begin{aligned}
  P_{\rm r} \approx P_{\rm A}f^2b^2 \Big|\sum_{n=1}^{N}e^{j(\varphi_n+\omega_n+\frac{2\pi d_n}{\lambda})}\Big|^2.
  \end{aligned}
 \end{equation}

\subsection{Layer-1 Codebook}
The proposed layer-1 codebook is designed to estimate $d$ with the information of the received power. Equation (\ref{received-power-approximation}) shows that $P_{\rm r}$ is determined by the UE location and the IRS beamforming together, which are coupled to determine the phase of the received signal. Therefore, there is no independent functional relationship between $P_{\rm r}$ and $d$. To tackle this issue, we propose to eliminate the influence of the UE location on the phase of the received signal by applying random phase codeword based beamforming to generate an omnidirectional IRS beam across the near-field region. 

Specifically, the proposed layer-1 codebook is composed of $C$ random-phase IRS beamforming vectors, which is denoted as:
\begin{equation}   
 \label{1_layer_codebook}
 \begin{split}
\left[ e^{j\mu_{c,1}}, e^{j\mu_{c,2}}, \cdots,e^{j\mu_{c,N}}\right],~~c = 1,\cdots,C,
 \end{split}
\end{equation}
where $\mu_{c,n}\overset{i.i.d}{\sim}{\mathcal{U}(0,2\pi)},~\forall~n\in \{1, 2, \cdots, N\},~c \in \{1, 2, \cdots, C\}$.
These $C$ beamforming vectors are designed to be executed for a limited period of time, which generates a cumulative omnidirectional IRS beam for the estimation of UE distance, which is elaborated on below.

According to (\ref{received-power-approximation}), the average received signal power of UE, corresponding to the proposed layer-1 codebook, is given by 
\vspace{-0.5em}
\begin{equation}   
 \label{UE_received_power}
 \begin{split}
 \Bar{P_{\rm r}} &= \frac{P_{\rm A}f^2b^2}{C}\sum_{c=1}^{C}
\underbrace{\Big|\sum_{n=1}^{N}e^{j(\overbrace{\mu_{c,n}+\omega_n+\frac{2\pi d_n}{\lambda}}^{\zeta_{c,n}})}\Big|^2}_{Q_c},
 \end{split}
\end{equation}
where $Q_c$ denotes the received signal power divided by $P_{\rm A}f^2b^2$ under the $c$-th codeword. $\zeta_{c,n}$ denotes the phase of the UE's received signal from the ${n}$-th IRS element under the $c$-th codeword.

\newtheorem{lemma}{Lemma}[]
\begin{lemma}[\textbf{Omnidirectivity of random-phase IRS beam}]
\label{lemma1}
For any UE location $(\theta, d)$, when ${C\to+\infty}$,
\begin{equation} 
 \label{omi_distribution}
\Bar{P_{\rm r}}\sim \mathcal{N}\Big(NP_{\rm A}f^2b^2, {N(N-1)P_{\rm A}^2f^4b^4}/{C}\Big),
\end{equation} 
and the unbiased estimation of $d$ can be expressed as
\begin{equation}   
 \label{estimation_of_l}
 \begin{split}
 \hat{d} &= \sqrt{\frac{G^{\rm U} A^{\rm U}G^{\rm I}NP_{\rm A}f^2}{4\pi \Bar{P_{\rm r}}}}.
 \end{split}
\end{equation}
\end{lemma}

\begin{proof}
For a given $(\theta, d)$, $d_n$ and $\omega_n$ are fixed due to the fixed deployment of the BS and IRS. According to (\ref{UE_received_power}), we have
\begin{equation}   
 \label{Q_c}
 \begin{split}
Q_c &= \Big|\sum_{n=1}^{N} e^{j\zeta_{c,n}}\Big|^2\\
&= \Big|\sum_{n=1}^{N}(\cos{\zeta_{c,n}} +i\sin{\zeta_{c,n}})\Big|^2\\
&= \Big(\sum_{n=1}^{N}\cos{\zeta_{c,n}}\Big)^2 +\Big(\sum_{n=1}^{N}\sin{\zeta_{c,n}}\Big)^2\\
&= N+
2\Big(\underbrace{
\sum_{1\leq i < j\leq N} \cos{(\zeta_{c,i}-\zeta_{c,j})}
}_{N(N-1)/2~items}\Big).
 \end{split}
\end{equation}
Since $\mu_{c,n}\overset{i.i.d}{\sim}{\mathcal{U}(0,2\pi)}$ and signal phase has a period of $2\pi$, we have $\zeta_{c,n}\overset{i.i.d}{\sim}{\mathcal{U}(0,2\pi)}$, and $(\zeta_{c,i}-\zeta_{c,j}) \overset{i.i.d}{\sim}{\mathcal{U}(0,2\pi)}, i\neq j$. Therefore, $\cos{(\zeta_{c,i}-\zeta_{c,j})},~ \forall~1\leq i < j\leq N,~c \in \{1, 2, \cdots, C\}$ are independent and identically distributed, with mean 0 and variance 0.5. So the mean and variance of $P_{\rm A}f^2b^2Q_c$ are $NP_{\rm A}f^2b^2$ and $N(N-1)P_{\rm A}^2f^4b^4$, respectively. 
Further, by applying the central limit law, (\ref{omi_distribution}) can be obtained when ${C\to+\infty}$, which proves the omnidirectional property of the proposed layer-1 codebook. The unbiased estimation of $d$ can be derived by substituting (\ref{h-approximation}) into (\ref{omi_distribution}), which yields (\ref{estimation_of_l}).
\end{proof}

Now we clearly show the omnidirectivity of the proposed random-phase beam pattern. Similar to (\ref{measured_signal_power}), the near-field beam patterns generated by the layer-1 codebook is represented as
\begin{equation}   
\label{layer-1-beam-pattern}
\begin{split}
\Bar{P}_{\rm p}(\mu_{c,n}) &= 
\frac{\frac{1}{C}\sum_{c=1}^{C}\Big|\sum_{n=1}^{N}\frac{1}{d_n^{'}}
e^{j(\mu_{c,n}+\omega_n+\frac{2\pi d_n^{'}}{\lambda})}\Big|^2}{\Big|\sum_{n=1}^{N}\frac{1}{d_n^{'}}\Big|^2},
\end{split}
\end{equation}
which is shown in Fig. \ref{fig:random_phase_beam_patterns} with $d = 6$~m, $C =$ 100, 500, 2000, respectively. The system setup is the same as that in Fig. \ref{fig:beam_patterns}. The theoretical random-phase beam pattern for ${C\to+\infty}$ is also illustrated. Fig. \ref{fig:random_phase_beam_patterns} shows that as $C$ increases, the mean value of the beam patterns remains the same, while the fluctuation of the beam patterns diminishes, which is in accordance with the conclusion of Lemma \ref{lemma1}.

 \begin{figure}[ht]
	\centering
	\includegraphics[scale = 0.5]{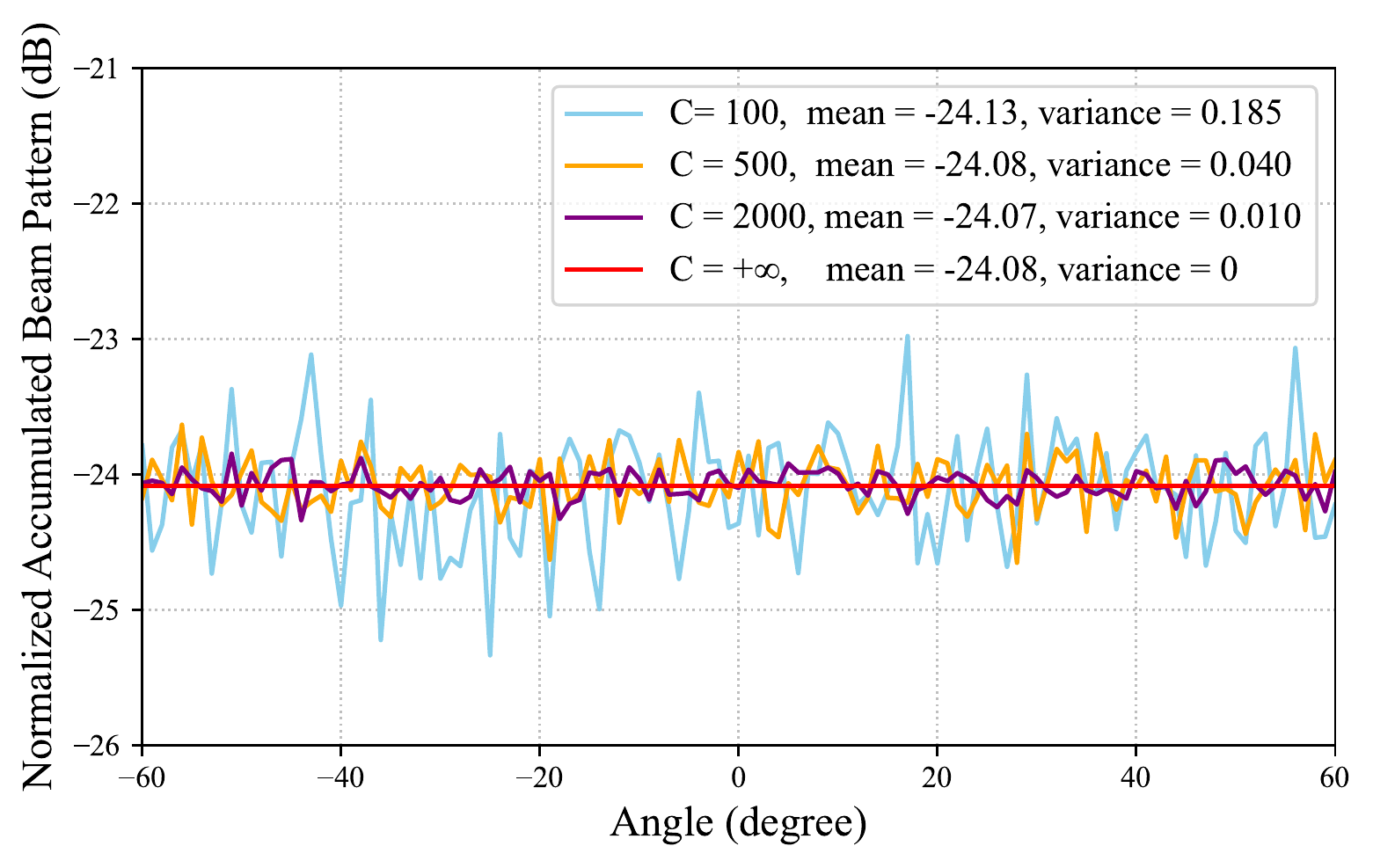}
	\caption{Normalized random-phase beam patterns.}
	\vspace{-1em}
	\label{fig:random_phase_beam_patterns}
\end{figure}

\newtheorem{remark}{Remark}
\begin{remark}[Implementation and Overhead of The Proposed Layer-1 Codebook]
\label{remark1}
\rm{Since only 1 omnidirectional IRS beam is generated and used in the proposed layer-1 codebook. In practice, the $C$ beamforming vectors in the layer-1 codebook can be pre-stored in the IRS controller, such as a Field Programmable Gate Array~(FPGA) controller, and performed within the time period of 1 traditional codeword, such as a DFT codeword, by means of phase switching at the nanosecond level\footnote{Based on the existing 5G standards\cite{3GPP.38.211}, one traditional codeword is executed during a synchronization signal and physical broadcasting channel block~(SSB), which spans 4 symbols and contains 960 resource elements~(REs) in total. Therefore, under the sampling rate required in the current wireless communication systems, the UE can successfully collect all the $C$, e.g., 500, changes of IRS beam pattern of the proposed layer-1 codebook.
In addition, in order to execute the layer-1 codebook~(e.g., $C = 500$) within one SSB~(e.g., sub-carrier of 60~kHz, which corresponds to the SSB period of 66.67 us), IRS hardware needs to achieve a phase switching cycle of 133.33 ns, which is achievable by existing FPGA controllers\cite{FPGA} and diodes\cite{diode}.}. 
Therefore, the training overhead of the layer-1 codebook is equivalent to that of only one DFT codeword. Then, the UE measures the average received power during this SSB period for the distance estimation of (\ref{estimation_of_l}). Equation
(\ref{omi_distribution}) shows that the variance of $\Bar{P_{\rm r}}$ is inversely proportional to $C$. Therefore, $C$ can be increased to enhance the omnidirectivity of the random-phase IRS beam and improve the accuracy of $\hat{d}$.}
\end{remark}


\subsection{Layer-2 Codebook}
The layer-2 codebook is designed to search for the UE angle, $\theta$, based on the estimated UE distance $\hat{d}$ in Section III-A. According to the compressed sensing-based analysis in \cite{Channel-Estimation-for-XL-MIMO}, the spatial angles corresponding to the DFT codebook are adopted as the candidate UE angles, which are represented as
\begin{equation} 
 \label{DFT-angles}
\theta_m = \arcsin{\frac{2m-N-1}{N}},~m = 1, 2, \dots, N.
\end{equation} 
The candidate UE locations are denoted as $(\theta_m, \hat{d}),~m = 1, 2, \dots, N$. The layer-2 codebook, composed of $N$ codewords, is generated by substituting $(\theta_m, \hat{d})$ into (\ref{euc_distance}) and (\ref{IRS_phase_config}), thus yielding
\begin{equation} 
 \label{layer-2-codebook}
 \begin{split}
      \left[ e^{-j(\omega_1 + \frac{2\pi d_1(\theta_m,\hat{d})}{\lambda})}, \cdots,e^{-j(\omega_N + \frac{2\pi d_n(\theta_m,\hat{d})}{\lambda})}\right],m = 1, \dots, N.  
 \end{split}
\end{equation} 
The layer-2 codewords are performed to generate $N$ beams with focuses on the $N$ candidate UE locations, respectively. Then, the UE measures the received power of each beam and reports the beam index corresponding to the maximal received power, which indicates the estimated UE location and the optimal IRS codeword. The procedure of the proposed two-layer beam training scheme is summarized in Algorithm \ref{beam-training-procedure}.
\begin{figure*}[ht]
    \centering
    \subfigure[Distance estimation error v.s. Reference SNR.]{
        \includegraphics[scale = 0.39]{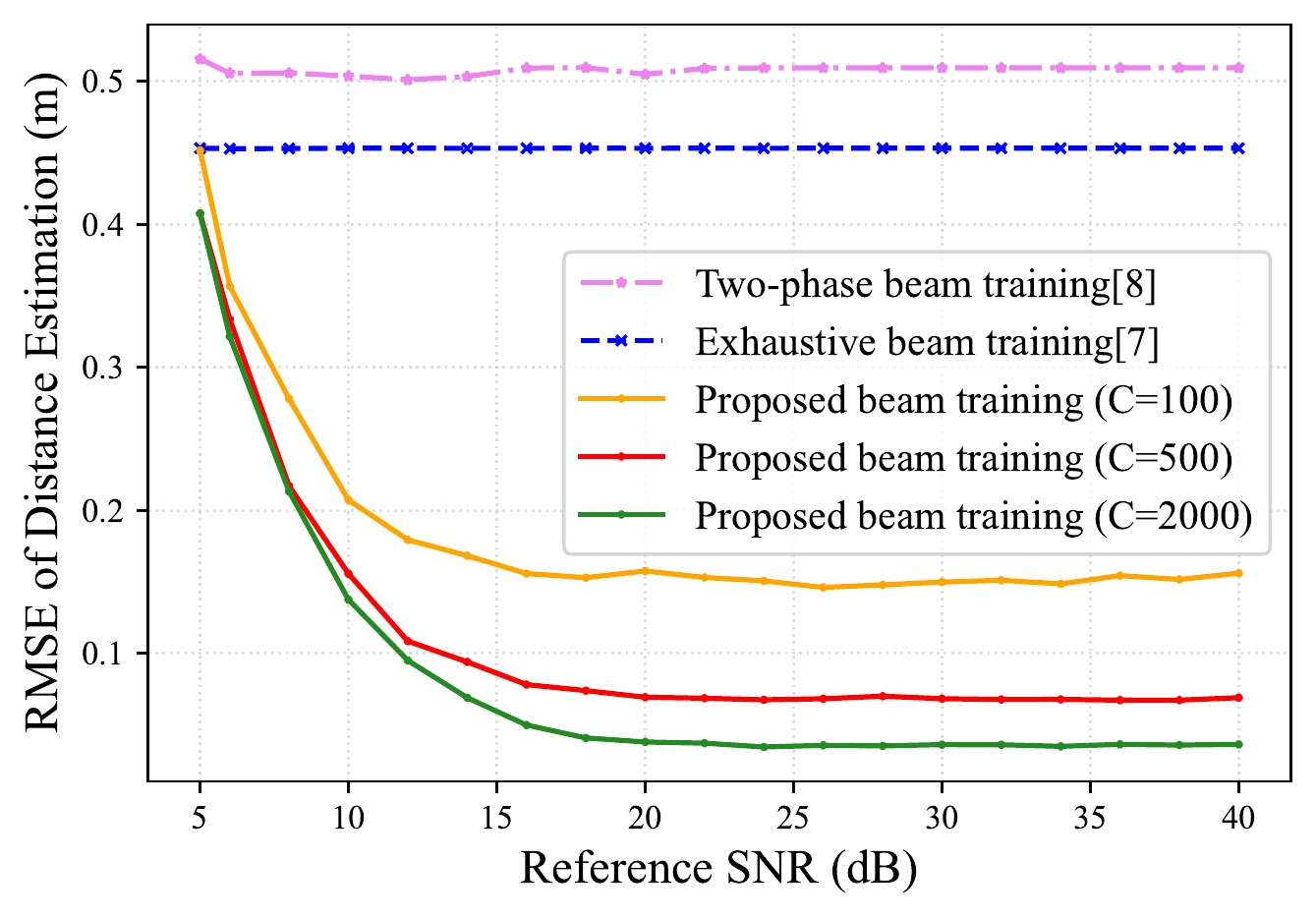}
	\vspace{-0.5em}
	\label{fig:distance-SNR}
    }
    \subfigure[Angle estimation error v.s. Reference SNR.]{
	\includegraphics[scale = 0.39]{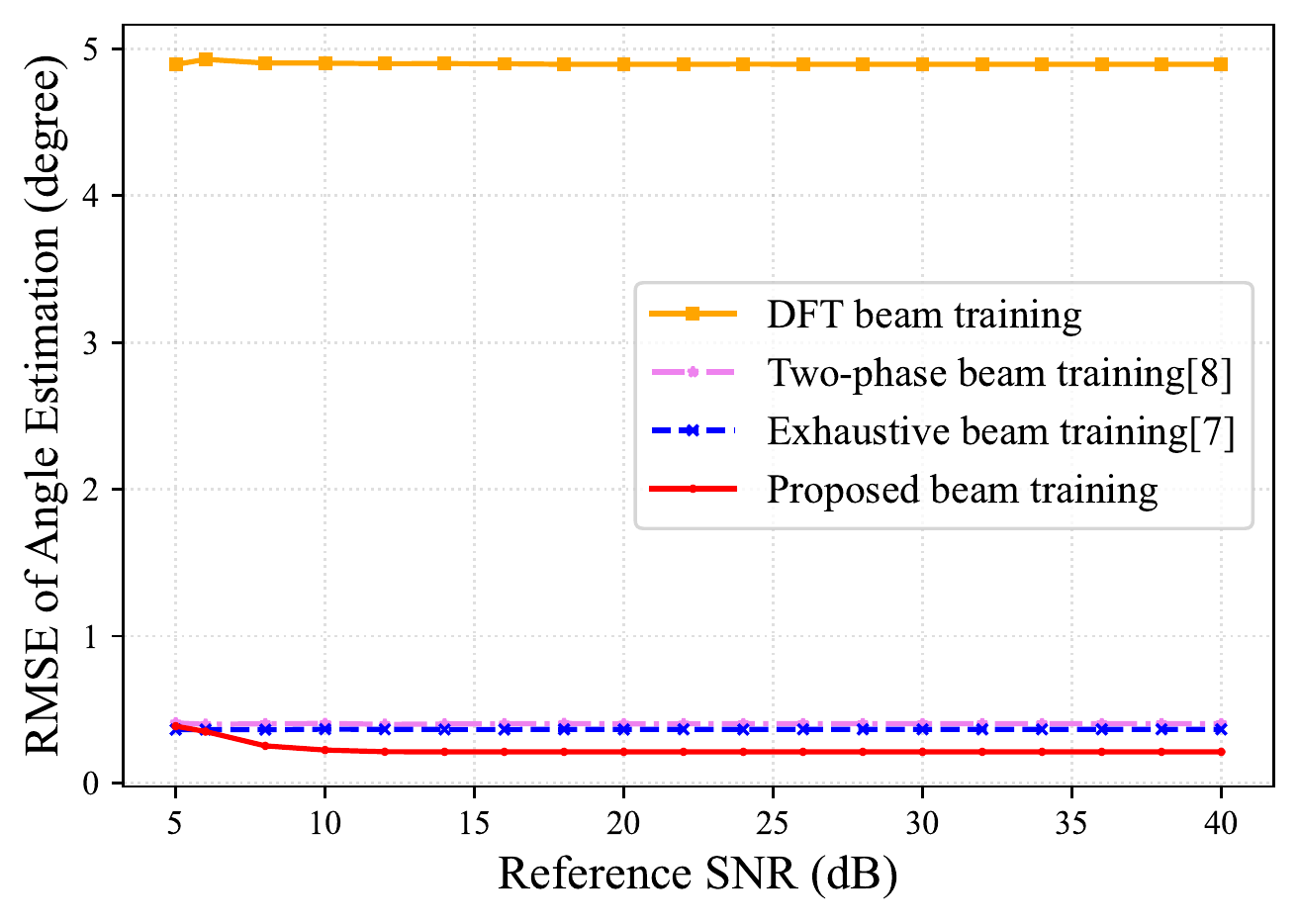}
	\vspace{-0.5em}
	\label{fig:angle-SNR}
    }
    \subfigure[Achievable rate v.s. Reference SNR.]{
    \includegraphics[scale=0.39]
    {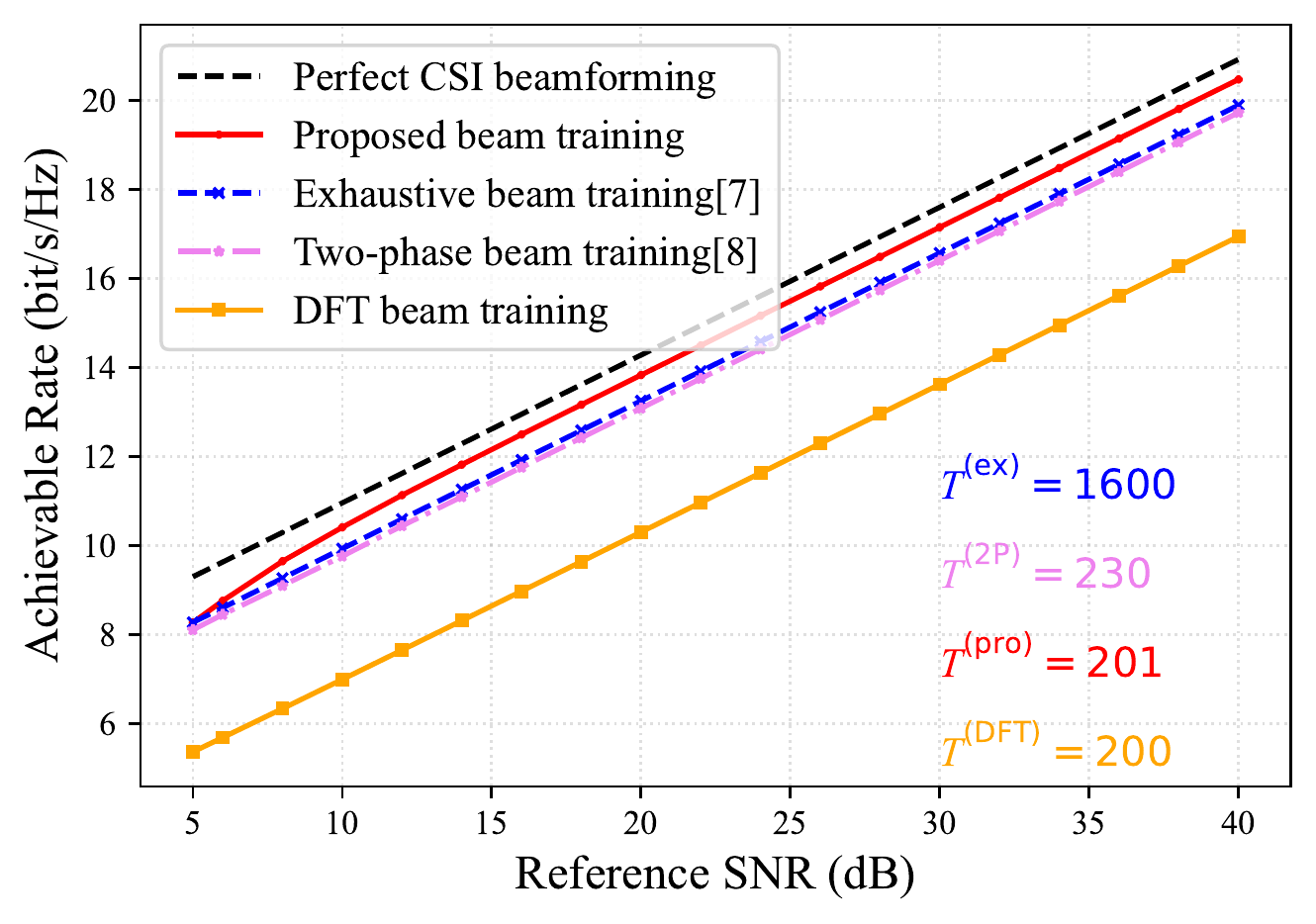}
\vspace{-0.5em}
\label{fig:rate-SNR}
    }
\quad    
    \subfigure[Distance estimation error v.s. User distance.]
    {\includegraphics[scale = 0.37]{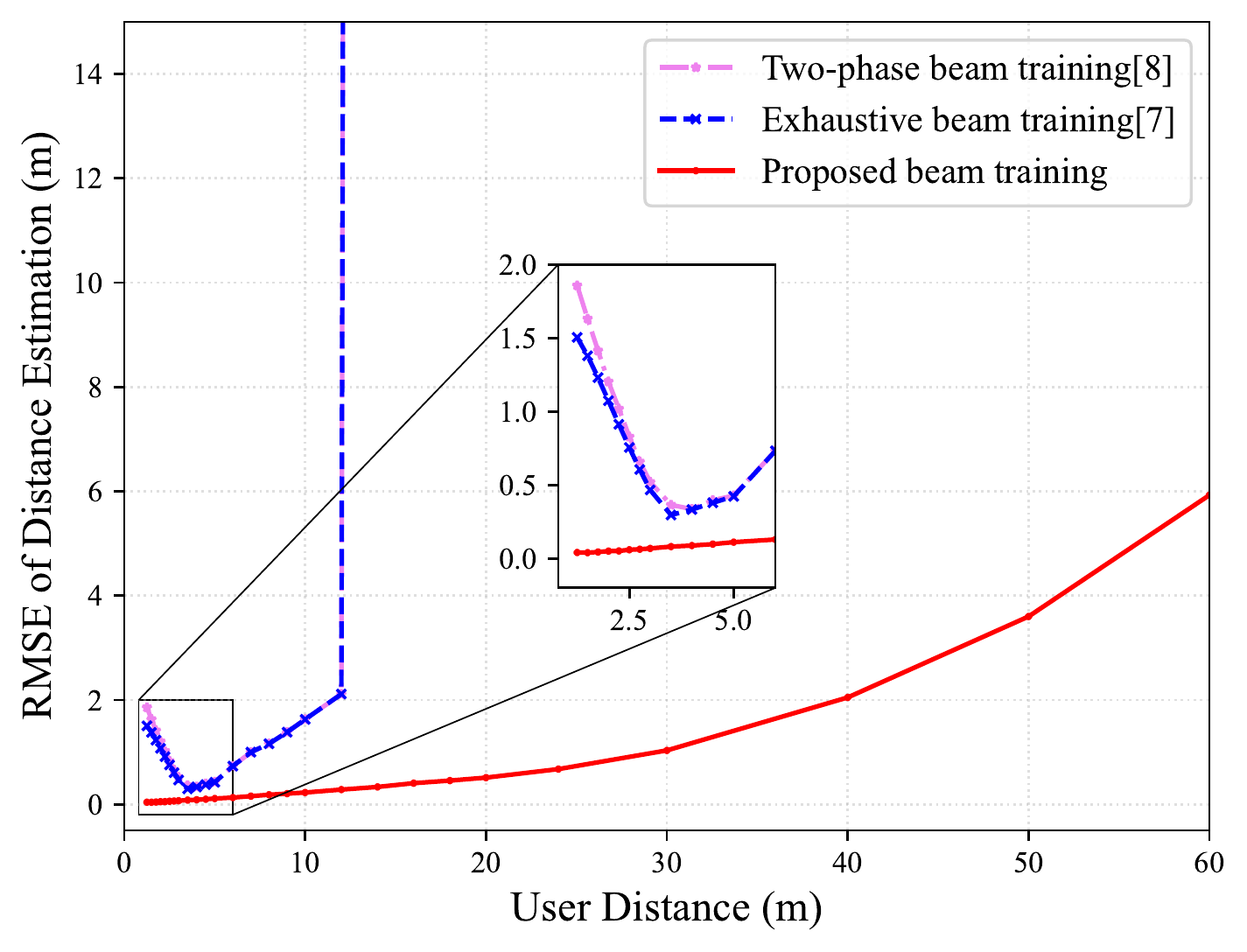}
	\vspace{-0.5em}
	\label{fig:distance-distance}
    }
    \subfigure[Angle estimation error v.s. User distance.]
    {\includegraphics[scale = 0.37]{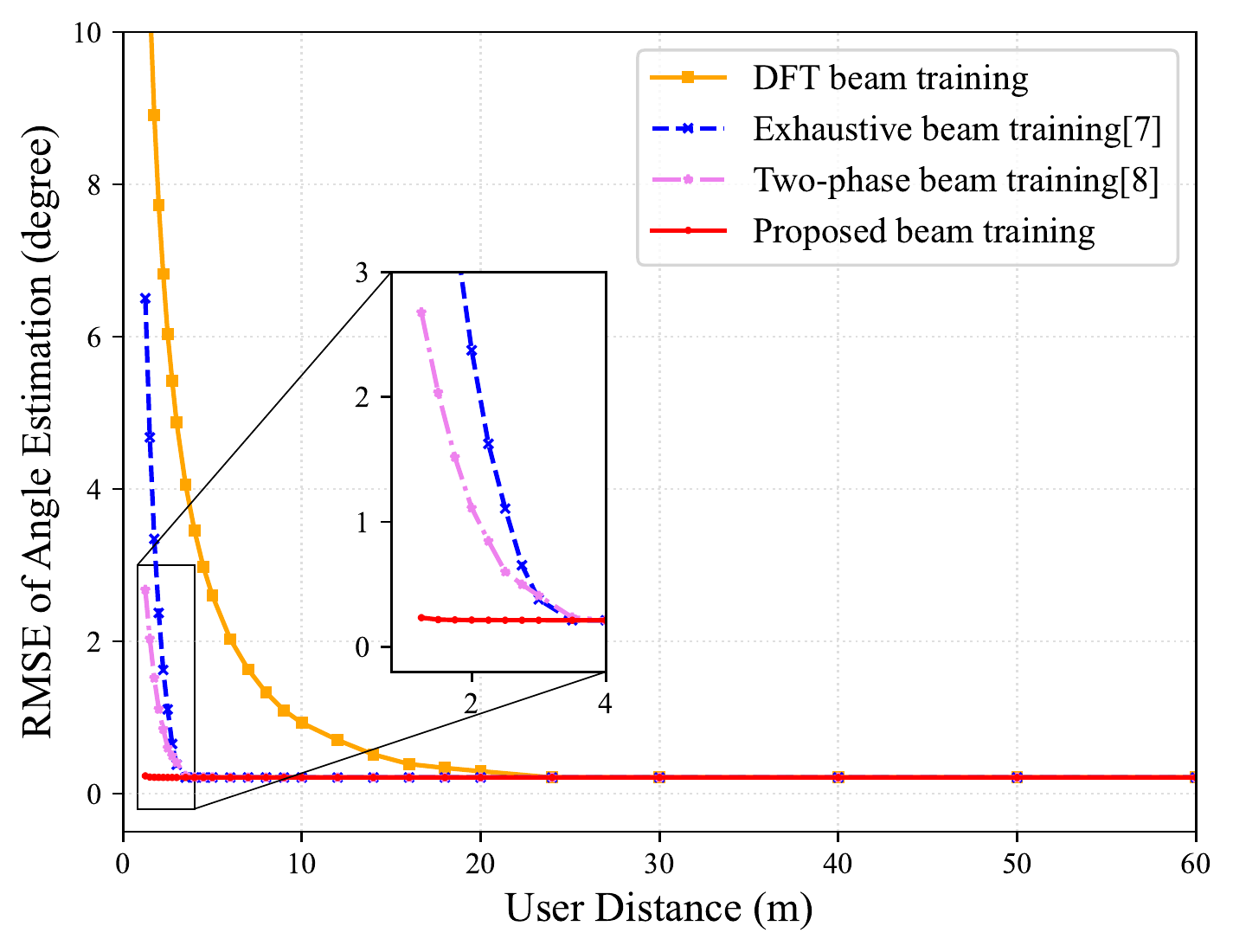}
	\vspace{-0.5em}
	\label{fig:angle-distance}
    }
    \subfigure[Achievable rate v.s. User distance.]{\includegraphics[scale=0.37]
    {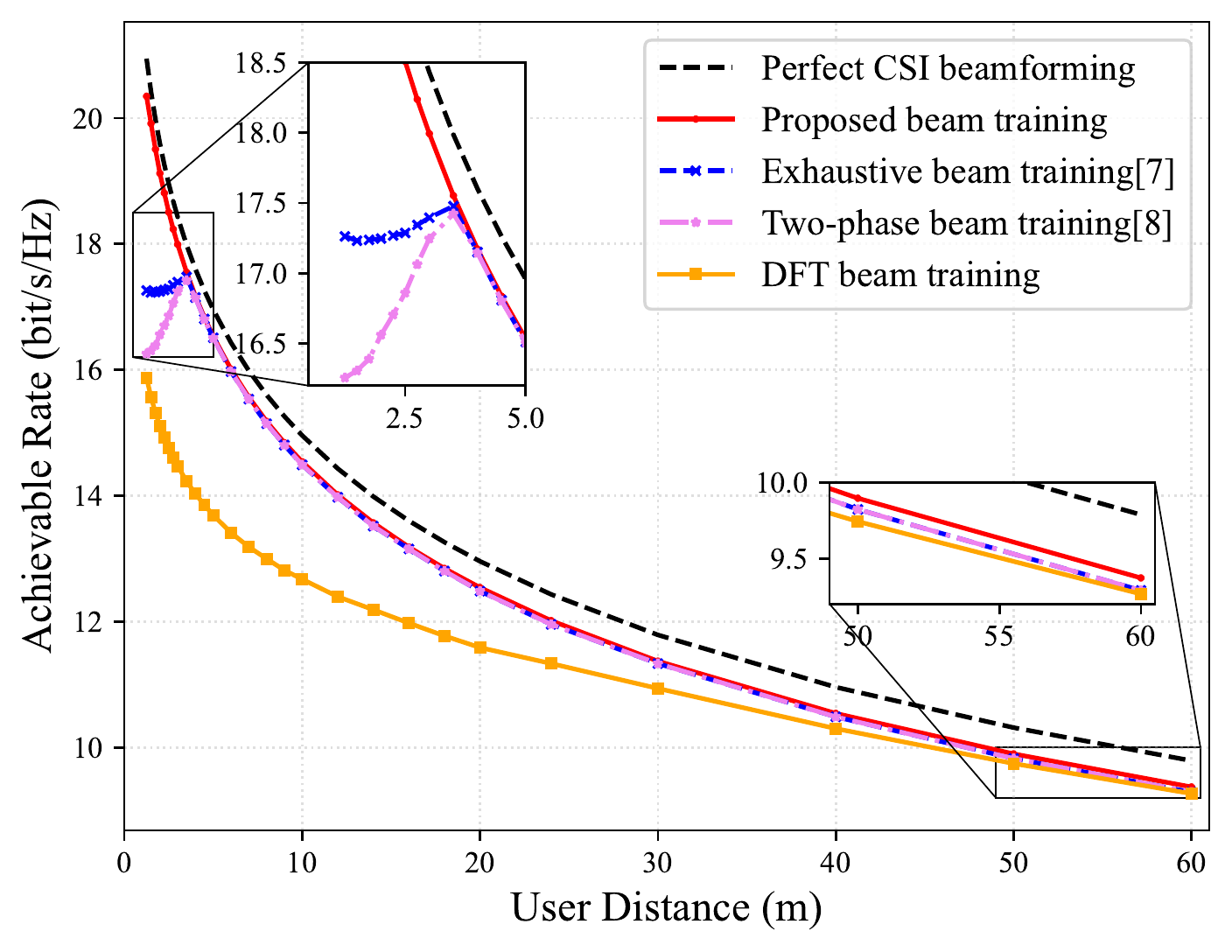}
\vspace{-0.5em}
\label{fig:rate-distance}
    }
    \caption{Performance comparison of the proposed two-layer codebook based beam training scheme with the exhaustive search based beam training scheme\cite{Channel-Estimation-for-XL-MIMO}, the two-layer beam training scheme\cite{Fast-Near-Field-Beam-Training-IRS}, and the DFT beam training scheme.}
    \vspace{-0.5em}
    \label{simulation-results}
\end{figure*}

\begin{algorithm}[ht]
	\caption{Proposed two-layer beam training procedure}
    \label{beam-training-procedure}
    \begin{algorithmic}[1] 
    \STATE \textbf{Layer 1: Distance estimation}
\STATE IRS performs the layer-1 codebook in (\ref{1_layer_codebook}) to generate 1 cumulative omnidirectional IRS beam. 
\STATE UE measures the average received power, $\Bar{P_r}$.
\STATE Obtain the distance estimation, $\hat{d}$, according to (\ref{estimation_of_l}).
    \STATE \textbf{Layer 2: Angle estimation}
    \STATE Obtain the layer-2 codebook according to (\ref{DFT-angles}) and (\ref{layer-2-codebook}).
    \STATE IRS performs the layer-2 codebook to generate $N$ sweeping beams.
    \STATE UE reports the beam index corresponding to the maximal received power.
    \STATE \textbf{Output:} The optimal IRS codeword.
    \end{algorithmic}
\end{algorithm}

\begin{remark}[Training overhead]
\label{remark3}
\rm{Since only 1 omnidirectional IRS beam is generated by the layer-1 codebook, and $N$ beams are generated by the layer-2 codebook, the overall training overhead of the proposed two-layer codebook based beam training is given by $T^{(\rm{pro})} = 1+N$, which is only 1 more than that of the DFT codebook based bean training scheme, $T^{(\rm{DFT})} = N$.
Moreover, $T^{(\rm{pro})}$ is much smaller than the training overhead of the exhaustive search based beam training scheme\cite{Channel-Estimation-for-XL-MIMO}, $T^{(\rm{ex})} = NS$, and that of the two-phase beam training scheme\cite{Fast-Near-Field-Beam-Training-IRS}, $T^{(\rm{2P})} = N+KS$, where $S$ and $K$ denote the number of candidate UE distances and the number of shortlisted candidate UE angles, respectively. To further reduce the overhead of the proposed layer-2 codebook, multi-beam and hierarchical strategy \cite{multi-beam-training-IRS} can be adopted, which is left for our future work.}
\end{remark}

\section{Numerical Results}

\begin{table}[ht]
	\centering  
	\caption{Parameter configuration}  
	\label{parameters}  
	\begin{tabular}{cc|cc}  
		        \hline  
		& & & \\[-6pt]  
		Parameters&Values&Parameters&Values \\  
		        \hline
		& \\[-6pt]  
		$N$&200&$\lambda$&0.01~m \\
		        \hline
    	& \\[-6pt]  
		$D$ & 1~m & $\sigma^2$ & -94~dBm \\
		        \hline
        & \\[-6pt]  
		$G^{\rm I},G^{\rm U}$&1&$A^{\rm U}$&$\lambda^2/(4\pi)$ \\
		        \hline
	\end{tabular}
\end{table}

A Monte-Carlo simulation is conducted to validate the effectiveness of the proposed two-layer codebook based beam training scheme. The channel model described in Section II is adopted, with parameters shown in Table \ref{parameters}. The reference SNR of the IRS assisted wireless system is defined as ${\rm SNR} = \frac{P_{\rm A}Nf^2G^{\rm I}G^{\rm U} A^{\rm U} }{4\pi d^2 \sigma^2}$, where the noise power is set as -94~dBm (corresponding to a bandwidth of 100~MHz). The achievable rate is calculated with equation (\ref{achievable-rate}). Additionally, the number of the independent channel realizations for each result point is set to 1000. Four benchmark schemes are adopted: 1) \emph{perfect-CSI based beamforming}, which provides the upper bound of rate performance; 2) \emph{exhaustive search based beam training} with $S = 10$ \cite{Channel-Estimation-for-XL-MIMO}; 3) \emph{two-phase beam training} with $K = 3, S = 10$ \cite{Fast-Near-Field-Beam-Training-IRS}; 4) \emph{DFT beam training}. 

Figs. 4(a)--4(c) show the effects of $\rm{SNR}$ on the Root Mean Squared Error~(RMSE) of distance estimation and angle estimation, and the achievable rate, where $d = 3$ m and $\theta$ is randomly distributed in [$-60^\circ$, $60^\circ$). Several observations can be made as follows. First, the RMSE of distance estimation decreases as $C$ increases. This is consistent with the conclusion in Remark 1.
Second, compared with the polar-domain beam training scheme in \cite{Channel-Estimation-for-XL-MIMO} and the two-phase beam training scheme in \cite{Fast-Near-Field-Beam-Training-IRS}, the proposed beam training scheme yields the least estimation errors of the UE distance and angle, thereby achieving better performance in terms of the achievable rate. 
Third, the training overhead of the proposed beam training scheme is significantly smaller than that of the polar-domain beam training scheme\cite{Channel-Estimation-for-XL-MIMO} and the two-phase beam training scheme\cite{Fast-Near-Field-Beam-Training-IRS} (i.e., 201 versus 1600 and 230). 
Moreover, it is observed that the proposed beam training scheme significantly outperforms the DFT codebook based beam training scheme under different SNRs, yet with a similar training overhead (i.e., 201 versus 200). 
Additionally, it is observed that the RMSE of distance estimation of the proposed beam training scheme increases in the very low SNR regime, leading to certain performance degradation. This problem can be readily solved by increasing the transmit power or the training overhead of the proposed layer-1 codebook (e.g., increasing the training overhead from 1 to 10 brings a 10~dB increase in the received signal energy).

Finally, we investigate the effects of UE distance $d$ on the RMSE of distance estimation and angle estimation, and the achievable rate in Figs. 4(d)--4(f), where $P_{\rm A}f^2$ is set to be 0.05 mW and $\theta$ is randomly distributed in [$-60^\circ$, $60^\circ$). In the range of 1-5~m, it is observed that the proposed beam training scheme provides accurate estimation of UE distance and angle, while the polar-domain beam training scheme and the two-phase beam training scheme yield significant estimation errors due to the beam broadening effect in the near-field region,
which results in an obvious decrease in achievable rate in this range. 
In the range of 5-60~m, it is observed that, on one hand, the RMSE of distance estimation of the proposed beam training scheme increases with the UE distance due to the decrease of SNR. On the other hand, the RMSE of distance estimation of the polar-domain beam training scheme and the two-phase beam training scheme increase much more rapidly because the number of distance samples is limited and the non-uniform distance sampling method results in distance estimation being always $+\infty$ in the long-distance region. 
Furthermore, the performance of the proposed beam training scheme outperforms those of the benchmarks, and is very close to perfect CSI beamforming under different UE distances, which is due to the limited angle-domain sampling step and the presence of noise.  
Last, the achievable rates of all beam training schemes monotonically decrease with the UE distance due to the more severe path-loss, and the achievable rates of all the schemes converge as the UE distance increases to satisfy the far-field assumption.




\section{Conclusion}

In this paper, we propose a novel two-layer codebook based beam training scheme to realise the near-field beam training of IRS with both low overhead and high IRS array gain. The proposed beam training scheme is inspired by the characteristics of the near-field beam patterns under near-field/far-field/random-phase beamforming, which enables efficient estimation of the distance and angle of the UE.
We analyze the effectiveness and overhead of the proposed beam training scheme to demonstrate its superiority over the existing benchmarks. Numerical results show that, compared with the benchmarks, the proposed beam training scheme provides more accurate estimation of the UE distances and angles, hence achieving close rate with the upper bound across the near-field region, yet with a smaller training overhead. Moreover, it is worth noting that the
proposed beam training scheme is general and can also be applied to IRS’s far-field region as well as BS’s transmit beam training without IRS.


\bibliographystyle{IEEEbib}
\bibliography{IEEErefs}
\end{document}